\documentclass[12pt,letterpaper]{article}
\usepackage{amsmath,amsfonts,amsthm,amssymb}

\swapnumbers 

\newtheorem{lemma}{Lemma}[section]

\newtheorem{theorem}[lemma]{Theorem}
\newtheorem{example}[lemma]{Example}

\theoremstyle{definition} 
\newtheorem{definition}[lemma]{Definition}

\begin{document}

	\title{Probability Measures and projections\\ on Quantum Logics}	

\author{O\v{l}ga N\'an\'asiov\'a\footnote{Slovak University of Technology, Faculty of Electrical Engineering and Information Technology, Institute of Computer Science and Mathematics, Bratislava, Slovakia, nanasiova@stuba.sk,  The first  author was supported by grant
		VEGA No. 1/0710/15}\hspace{.05 in},   Viera  \v{C}er\v{n}anov\'a\footnote{Department of Mathematics and Computer Science, Faculty of Education, Trnava University, Trnava, Slovakia, vieracernanova@hotmail.com}\hspace{.05 in},   \v{L}ubica Val\'a\v{s}kov\'a\footnote{Slovak University of Technology, Faculty of Civil Engineerig, Department of Mathematics and Descriptive Geometry, Slovakia, lubica.valaskova@stuba.sk, the third author was supported by grant VEGA 1/0420/15	  }}

\date{}
\maketitle

\begin{abstract}
	The present paper is devoted to modelling of a probability measure of logical connectives on a quantum logic (QL), via a $G$-map, which is a special map on it.
	We follow the work in which   the probability of logical conjunction, disjunction and symmetric difference and their negations for non-compatible propositions are studied.
	
We study such a $ G $-map on quantum logics, which is a probability  measure of a  projection and   show, that unlike classical (Boolean) logic, probability measure of projections on a quantum logic are not necessarilly pure projections. 
			
  	We  compare  properties of a  $G$-map on QLs with properties of a  probability measure related to logical connectives on a Boolean algebra.

 \textbf{ Keywords}: logical connectives, orthomodular lattice, quantum logic, probability measure, state
\end{abstract}

   	\section*{Introduction }

 The problem of  modelling of probability measures for logical connectives of  non-compatible propositions  started by publishing the paper  Birkhoff, G.,  von Neumann, J. \cite{BirkhoffNeumann}.    Quantum logic allows to model situations with non-compatible events (events that are not simultaneously measurable). Methods of quantum logic appear in data processing, economic models, and in other domains of application e.g. \cite{BirkhoffNeumann,sozz,khrenn,cluster}.
  
  Calculus for non-compatible observables has been described in \cite{Kalina-Nanasiova}, while modelling of logical connectives in terms of their algebraic properties and algebraic structures can be found in \cite{tolo-syl-AND, Herman-Marsden-Piziak, pavicic2016}.

  The present paper follows up the work \cite{NanVal}, 
  where the authors studied logical connectives: conjuction, disjunction, and symmetric difference together with their negations, from the perspective of a~probability measure. An overview of various insights into this issue is provided in   \cite{jarek}.

  The paper is organized as follows. Section \ref{basic def} reminds some basic notions and their properties. A special function that associates a~probability measure to some logical connectives on a quantum logic is defined and studied in Section \ref{Special maps} and Section \ref{projection}. In the last Section \ref{Summary}  properties of a  $G$-map are compared with properties of a  probability measure related to logical connectives on a Boolean algebra.
   	\section{\label{basic def} Basic definitions and properties}
In the first part of this section, we recall fundamental notions: orthomodular lattice, compatibility, orthogonality, state, and their basic properties.
For more details, see \cite{tolo-syl, pulman}.
In the second subsection, we recall some situations with two-dimensional states allowing to model a probability measure of logical connectives in the case of non-compatible events \cite{Kalina-Nanasiova},\cite{nanasiova1}-\cite{Nan-Val-triangle},\cite{Pykacz Nan Val}. 

\subsection{Quantum logic}
 \begin{definition}\label{OcL} An orthomodular lattice (OML) is a lattice $L$
	with $0_L$ and $1_L$ as the smallest and the greatest element,
	respectively, endowed with a unary operation $a\mapsto a'$ that satisfies:
	(i)	$a'':=(a')'=a$; (ii) $a\leq b$ implies $b'\leq a'$; (iii) $a\vee a'=1_L$; 
	(iv) $a\leq b$ implies $b=a\vee(a'\wedge b)$  (the orthomodular law).
	\end{definition}    

\begin{definition}
	Elements $a,b$ of an orthomodular lattice $L$ are called \\ 
-- orthogonal if $a\leq b'$;  (notation $a\perp b$ );\\
--  compatible if $a=(a\wedge b)\vee(a\wedge b')$; (notation $a\leftrightarrow b$).
\end{definition}

\begin{definition} 	A state on an OML $L$ is a function $m:L\to [0,1]$ such that\\
(i) $m(1_L)=1$;\\
(ii) $a\perp b$ implies $m(a\vee b)=m(a)+m(b)$.
\end{definition}

Note that the notions {\emph{state}} and {\emph{probability measure}} are closely tied, and it is clear that $m(0_L)=0$. 

There exist three kinds of OMLs: without any state, with exactly one state and with  infinite number of states (see e.g. \cite{pavicic-stavy a ql}). The first and the second type  of OLMs as a basic structure are not suitable to build a generalized probability theory. The last type of OMLs, which has infinite number of states is considered in the present paper.  

 \begin{definition}	\label{QL} An OML $L$ with   infinite number of states is  called a quantum logic (QL).
 \end{definition}

When studying   states on a quantum logic, one can   meet  some problems, that do not  exist on a Boolean algebra. It means, that some of basic properties of probability measures are not necessarilly satisfied for non-compatible random events. Here are some of them: Bell-type inequalities  (e.g. \cite{khrenn,khrenn2,Pitovsky,Pykacz Nan Val}), Jauch-Piron state, (e.g. \cite{jauch-piron states,honza}), problems of pseudometric (see \cite{NanVal}).

\subsection{Probability  measures of logical connectives on  QLs }
In the paper \cite{nanasiova2},  the notion of \emph {a map for simultaneous measurements (an s-map)} on a  QL  has been introduced. This function is a measure of conjunction even for non-compatible propositions, see e.g. \cite{jarek}.
\vskip 0.5pc
A map $p:L\times L \to [0,1]$ is called a \emph{map for simultaneous measurements} (abbr. \emph{s-map})  if the
following conditions hold:\\
(s1) 	$p(1_L,1_L)=1$;\\
(s2) if $a\perp b$ then $p(a,b)=0$;\\
(s3) if $a\perp b$ then for any $c\in L$:

 $ p(a\vee b,c)=p(a,c)+p(b,c)$ and  $\,\,p(c,a\vee b)=p(c,a)+p(c,b).$\\

The following properties of $s$-map have been proved:\\  Let $p:L\times L \to [0,1]$ be an $s$-map and $a,b,c \in L$. Then\\
1.	if $a\leftrightarrow b$ then $p(a,b)=p(a\wedge b,a\wedge b)=p(b,a)$;\\
2. 	if $a\leq b$ then $p(a,b)=p(a,a)$;\\
3. if $a\leq b$ then $p(a,c)\leq p(b,c)$ and $p(c,a)\leq p(c,b)$ for any $c\in L$;\\
4. $p(a,b)\leq \min{(p(a,a), p(b,b))}$;	\\
5. the map $m_p:L\to [0,1]$ defined as $m_p (a)=p(a,a)$ is a state  on $L$, induced by $p$.

The property 1.  shows that $s$-maps can be seen as providing probabilities of `virtual' conjunctions of propositions, even non-compatible ones, for in the case of compatible propositions the value $p(a,b)$ coincides with the value that a state $m_p$ generated by $p$ takes on the meet $a\wedge b$, which in this case really represents conjunction of $a$ and $b$ (\cite{jarek}).

On the other hand, the identity $p(a,b)=p(b,a)$  may not be true in general. 
So an $s$-map can be used for describing of stochastic causality \cite{Kalina-Nanasiova,Nan-Khre-marg}.
Moreover,   for any $a\in L$:    $\,\,m_p (a) = p(a,a) = p(1_L,a) = p(a,1_L)$. 

Measures of logical connectives  disjunction (\emph{$j$-map}) and symmetric difference (\emph{$d$-map}) are studied on a QL e.g. \cite{NanVal,Nan-Val-Dr}.
\vskip 0.5pc
Let $L$ be a QL. A map $q:L\times L \to [0,1]$ is called a \emph{join map} (abbr. \emph{j-map}) if the
following conditions hold:\\
(j1) $q(0_L,0_L)=0$,\quad  $q(1_L,1_L)=1$;\\
(j2) if $a\perp b$ then $q(a,b) = q(a,a) +q(b,b);$\\
(j3) if $a\perp b$ then for any $c\in L$: 

$q(a\vee b,c)=q(a,c)+q(b,c)-q(c,c)$ and 
$q(c,a\vee b)=q(c,a)+q(c,b)-q(c,c).$
\vskip 0.5pc 
If $p$ is an $s$-map on a QL, $m_p$ is a state induced by $p$ and $q_p:L\times L\to [0,1]$ such that for any $a,b\in L$  $q_p (a,b) =  m_p(a) + m_p(b) - p(a,b)$, then  $q_p$ is  a $j$-map~\footnote{It is easy to see that if $a\leftrightarrow b$, then $q_p(a,b)=m_p(a)+m_p(b)-m_p(a\wedge b)=m_p(a\vee b)$ which explains its name.}. 
\vskip 0.5pc
Let $L$ be a QL. A map $d: L\times L \to [0,1]$  
 is  called  a \emph{difference map} or simply \emph{d-map}\footnote{If $a\leftrightarrow b$, then  $d(a,b)=
 	m_d(a\bigtriangleup b)=m_d(a\wedge b')+m_d(a'\wedge b)$, where $m_d$ is a state induced by $d$. } 
 if  the following conditions hold:\\
(d1) $d(a,a)=0$  for any $a\in L$, and $d(1_L,0_L)=d(0_L,1_L)=1$;\\
(d2)  if  $a\perp b$ then $d(a,b)=d(a,0_L)+d(0_L,b)$;\\
(d3)  if $a\perp b$ then for any $c \in L$:  

$d(a\vee b,c)=d(a,c)+d(b,c)-d(0_L,c)$ and $d(c,a\vee b)=d(c,a)+d(c,b)-d(c,0_L).$

\section{\label{Special maps} Special bivariables maps on QLs }

In \cite{NanVal}, special bivariables maps $G$ satisfying $G(0_L,1_L)=G(1_L,0_L)$ have been introduced. The following defininition  brings an extended version of $G$-map.    
\begin{definition}\label{G-map} Let $L$ be a QL. A map $G:L\times L\to [0,1]$ is called a \emph{$G$-map} if the following holds:
	\begin{itemize}
		\item[(G1)] if $a,b\in\{0_L, 1_L\}$ then $G(a,b)\in\{0,1\}$;
		\item[(G2)] if  $a\perp b$ then $G(a,b)=G(a,0_L)+G(0_L,b)-G(0_L,0_L)$;  
		\item[(G3)] if $a\perp b$ then for any $c \in L$:  \\
		$$G(a\vee b,c)=G(a,c)+G(b,c)-G(0_L,c)$$ 
		$$G(c,a\vee b)=G(c,a)+G(c,a)-G(c,0_L).$$
	\end{itemize}
\end{definition}
A $G$-map enables modelling probability of  logical connectives even for non-compatible propositions.
\begin{lemma}\label{komp} Let  $G:L\times L\to [0,1]$ be a $G$-map, where $L$ is a QL. Then for $a\leftrightarrow b$ it holds 
	$$G(a,b)=G(a\wedge b,a\wedge b)+G(a\wedge b',0_L)+G(0_L,a'\wedge b)-2G(0_L,0_L).$$
\end{lemma}

\begin{proof}
	Consider $a,b$ compatible. Then $a=(a\wedge b)\vee (a\wedge b')$ and 
	$b=(a\wedge b)\vee (a'\wedge b)$.
	Since $a\wedge b$, $a\wedge b'$ are orthogonal, it follows immediately from Definition \ref{G-map}:
	\begin{eqnarray*}
	G(a,b)&=&G\left((a\wedge b)\vee (a\wedge b'),\,b\right)= G\left(a\wedge b,\,b\right)+G\left(a\wedge b',\,b\right)-
	G(0_L,\, b)\\
	&=&\, G\left(a\wedge b,\,a\wedge b\right)+G\left(a\wedge b,\,a'\wedge b\right)-G(a\wedge b,\,0_L)\\
	&&+ G\left(a\wedge b',\,0_L\right)+G\left(0_L,\,b\right)-G(0_L,\, 0_L)- G\left(0_L,\,b\right)\\
	&=&\, G\left(a\wedge b,\,a\wedge b\right)+G\left(a\wedge b,\,0_L\right)+
	G\left(0_L,\,a'\wedge b\right)-G\left(0_L,\,0_L\right)\\
	&&-G(a\wedge b,\,0_L)+ G\left(a\wedge b',\,0_L\right)-G\left(0_L,\,0_L\right)\\
	&=&\,G(a\wedge b,a\wedge b)+G(a\wedge b',0_L)+G(0_L,a'\wedge b)-2G(0_L,0_L).
		\end{eqnarray*}  (Q.E.D.)
\end{proof}

\begin{theorem}\label{1-G}
	Let  $G:L\times L\to [0,1]$ be a $G$-map, where $L$ is a QL. Then the map $G'=1-G$ is a $G$-map.	
\end{theorem}	
\begin{proof}
	It suffices to verify the rules (G2) and (G3) from Definition \ref{G-map}, because (G1) is obvious.\\
	(G2): Consider $a,b\in L$ such that $a\perp b$. Then
	\begin{eqnarray*}
	G'(a,b)&=&1-G(a,b)
	=1-\left(G(a,0_L)+G(0_L,b)-G(0_L,0_L)\right)\\
	&=&1-G(a,0_L)+1-G(0_L,b)-1+G(0_L,0_L)\\
	&=&G'(a,0_L)+G'(0_L,b)-G'(0_L,0_L).
	\end{eqnarray*}	
	\noindent	(G3): Consider $a,b,c\in L$, where $a\perp b$. Then
	\begin{eqnarray*}
	G'(a\vee b,c)&=&1-G(a\vee b,c)
	=1-\left(G(a,c)+G(b,c)-G(0_L,c)\right)\\
	&=&1-G(a,c)+1-G(b,c)-1+G(0_L,c)\\
	&=&G'(a,c)+G'(b,c)-G'(0_L,c).
	\end{eqnarray*}
The proof of the  second identity is similar. (Q.E.D.)
\end{proof}

There are sixteen families $\Gamma_i$, ($i=1,...,16$) of maps $G$ according to values in vertices $(1_L,1_L)$, $(1_L,0_L)$, $(0_L,1_L)$, $(0_L,0_L).$ Eight of them with $G(1_L,0_L)=G(0_L,1_L)$ are studied in \cite{NanVal}. See Table \ref{tab1-8}.  
\begin{table}[h] 
	\renewcommand{\arraystretch}{1.3}
	\caption{\label{tab1-8} \small{ Results from the paper \cite{NanVal}}  }
	\begin{center}
		\begin{tabular}{|l|l|l|l|l|l|l|l|l|} 	\hline	
	&$\,\Gamma_1\,$& $\Gamma_2$&$\Gamma_3$& $\Gamma_4$ &$\Gamma_5$ &$\Gamma_6$& $\Gamma_7$ &$\Gamma_8$\\
		\hline
		$G(0_L,0_L)$	&0& 0&0& 0 &1 &1&1&1\\
		\hline
		$G(1_L,0_L)$	&0& 0&1& 1 &1 &0&0&1\\
		\hline	
		$G(0_L,1_L)$	&0& 0&1& 1 &1 &0&0&1\\
		\hline	
		$G(1_L,1_L)$	&0& 1&1& 0 &0 &0&1&1\\
		\hline
	 prob. of	& $0_L$& $a\wedge b$  &$a \vee b$& $(a\Leftrightarrow b)'$ &$a'\vee b'$ &$a'\wedge b'$&$a\Leftrightarrow b$&$\, 1_L\,$\\
	 	& &  $a\leftrightarrow b$  &$a\leftrightarrow b$& $a\leftrightarrow b$ &$a\leftrightarrow b$ &$a\leftrightarrow b$&$a\leftrightarrow b$&\\
	\hline
\end{tabular}
\end{center}
\end{table}

Family $\Gamma_2$ is the set of all $s$-maps (measures of conjuntion), $\Gamma_3$ the set of all $j$-maps (measures of disjunction), and  $\Gamma_4$ is that of all $d$-maps (measures of symmetric difference) on a QL (see  \cite{NanVal} for more details).
In the present paper  a map $G$  generating a measure of projection on a QL is studied.

\section{\label{projection}Probability measures of projections on  QLs}

This part is devoted to  $\Gamma_9-\Gamma_{12}$ with values in the vertices shown in the Table \ref{tab9-16}. As  $\Gamma_9$ and $\Gamma_{10}$ are analogical, only $\Gamma_9$ is studied in detail. Moreover: $G\in \Gamma_{11}$ iff $1-G\in\Gamma_{9}$, and 
$G\in \Gamma_{12} $ iff  $1-G\in \Gamma_{10}$ (Proposition \ref{1-G} and  Table \ref{tab9-16}).

\begin{table}[h] 
	\renewcommand{\arraystretch}{1.3}
	\caption{\label{tab9-16}
		\small{$\Gamma_9-\Gamma_{12}$ values in vertices}}
	\centering
	\begin{tabular}{|c|c|c|c|c|c|c|c|c|} 	\hline	
		$ $	&$\,\Gamma_9\,$& $\,\Gamma_{10}\,$&$\,\Gamma_{11}\,$& $\,\Gamma_{12}\,$ \\
		\hline
		$\,G(0_L,0_L)\,$	&0     & 0            &1            & 1        \\
		\hline
		$G(0_L,1_L)$	&0& 1&1& 0 \\
		\hline
		$G(1_L,0_L)$	&1& 0&0& 1 \\
		\hline	
		$G(1_L,1_L)$	&1& 1&0& 0  \\
		\hline	
	\end{tabular}
\end{table}
\begin{lemma}\label{vlproje1} Let  $L$ be a QL and $G \in \Gamma_9$. Then for any $a,b\in L$ it holds
	\begin{enumerate}
		\item $G(1_L,a)=1$, $G(0_L,a)=0$;
		\item $G(a,0_L)=G(a,a)=G(a,1_L)$;
		\item $G(a,0_L)=\frac{1}{2}(G(a,b)+G(a,b'))$;
		\item $G(a,0_L)=\frac{1}{n}\sum_{i=1}^nG(a,b_i),$   where $b_1,\cdots ,b_n$ is an
		orthogonal partition of unity $1_L$.
	\end{enumerate}
\end{lemma} 
\begin{proof}
\begin{enumerate}
	\item  Let $G\in \Gamma_9$ and $a \in L$. Then from
	\begin{eqnarray*}
	1&=&G(1_L,1_L)=G(1_L,a)+G(1_L,a')-G(1_L,0_L)
	=  G(1_L,a)+G(1_L,a')-1\\
	2&=&G(1_L,a)+G(1_L,a')\\
	0&=&G(0_L,1_L)=G(0_L,a)+G(0_L,a')-G(0_L,0_L)
	=G(0_L,a)+G(0_L,a'),
	\end{eqnarray*}
	and taking into account that $G:L^2\to [0,1]$, it follows that $G(1_L,a)=1$ and $G(0_L,a)=0$.
	\item 	Let $G\in \Gamma_9$ and $a \in L$. Then
	\begin{eqnarray*}
	G(1_L,0_L)&=& G(a,0_L)+G(a',0_L)-G(0_L,0_L)=G(a,0_L)+G(a',0_L)\\
	G(1_L,a) &=& G(a,a)+G(a',a)-G(0_L,a)\\
	&=&G(a,a)+G(a',0_L)+G(0_L,a)-G(0_L,0_L)
	\end{eqnarray*}
	As $G(1_L,0_L)=G(1_L,a)$ and $G(0_L,a)=G(0_L,0_L)=0$, one obtains $G(a,0_L)=G(a,a)$.
	Moreover,
	\begin{eqnarray*}
	G(a,1_L)&=&G(a,a)+G(a,a')-G(a,0_L)\\
	&=&G(a,a)+G(a,0_L)+G(0_L,a')-G(0_L,0_L)-G(a,0_L)\\
	&=&G(a,a).
	\end{eqnarray*}
	\item Let $G\in \Gamma_9$ and $a,b \in L$. Then
	\begin{eqnarray*}
	G(a,1_L)&=&\,G(a,b)+G(a,b')-G(a,0_L)\\
	G(a,1_L)+G(a,0_L)&=&\,G(a,b)+G(a,b')\\ 
	2G(a,0_L)&=&\,G(a,b)+G(a,b').
	\end{eqnarray*}
	\item To prove this identity, it suffices to set $1_L=\bigvee_{i=1}^n b_i $  and follow the method used in the previous item. 
(Q.E.D.)
\end{enumerate} 
\end{proof}
\begin{theorem}\label{vlproje2} Let  $L$ be a QL, and $G \in \Gamma_9$. Then for any $a,b\in L$ it holds
	\begin{enumerate}
		\item If $a\leftrightarrow b$
			 then $G(a,b)=G(a,0_L).$
		\item For any choice of $b$, the map $m_b:L\to [0,1]$: $m_b(a)=G(a,b)$ is a state on $L$. 
	\end{enumerate}
\end{theorem}
\begin{proof}
Let $G\in \Gamma_9$. 
\begin{enumerate}
	\item Consider $a,b \in L, a\leftrightarrow b$. Then according to Lemmas \ref{komp} and \ref{vlproje1}
	\begin{eqnarray*}
	G(a,b)&=&G(a\wedge b,a\wedge b )+G(a\wedge b',0_L )+G(0_L,a'\wedge b )-2G(0_L,0_L)\\
	G(a,b)&=&G(a\wedge b,0_L )+G(a\wedge b',0_L )-G(0_L,0_L)\\
	G(a,b)&
	=&G(a,0_L ).
	\end{eqnarray*}
	\item Consider $b\in L$, and define $m_b:L\to [0,1]$: $m_b(a)=G(a,b)$. Then
	$$ m_b(1_L)=G(1_L,b)=1,$$
	and for  arbitrary $x,y\in L$, $x\perp y$, it holds $$m_b(x\vee y)=G(x\vee y,b)=G(x,b)+G(y,b)-G(0_L,b)=m_b(x)+m_b(y).$$
	Therefore $m_b$ is a state on L.  (Q.E.D.)
\end{enumerate} 
\end{proof}

From Proposition \ref{vlproje2} it follows that any $G\in \Gamma_9$ 
is a probability measure of the projection onto the first coordinate. Anological properties are fullfiled for any   $G\in \Gamma_{10}$, which is a probability measure  of the projection onto the second coordinate. 

If $L$ is a  Boolean algebra, then for any $G\in\Gamma_9$ ($G\in\Gamma_{10}$) it holds $G(a,b)=G(a,0_L)$ ($G(a,b)=G(0_L,b)$) for all $a,b\in L$. If $L$ is a QL but not a Boolean algebra, then the equality does not hold in general, as illustrates the following example.  	

\begin{example}\label{existenciaGPpr} 
Consider $L=\{0,1,a,a',b,b'\}$, a horizontal sum of Boolean algebras  $\mathcal{B}_a=\{0,1,a,a'\}$, $\mathcal{B}_b=\{0,1,b,b'\}$. Consider $r_1,r_2,u_1,u_2\in [0,1]$. Every $G\in \Gamma_9$ can be fully defined by Table \ref{tabG-proj}, where $\alpha=\frac{1}{2}(r_1+r_2)$, $ \beta=\frac{1}{2}(u_1+u_2)$ according to Lemma \ref{vlproje1}. If $r_1\neq r_2$ then $G(a,b)\neq G(a,0_L)$.
\begin{table}[h!] 
	\renewcommand{\arraystretch}{1.3}
\caption{\label{tabG-proj}
	\small{  $G$-maps from  $\Gamma_{9}$ on a horizontal sum of Boolean algebras} 
}
\centering
	\begin{tabular}{|c|c|c|c|c|c|c|} 	\hline	
			&$a$& $a'$&$b$& $b'$ &$0_L$ &$1_L$\\
		\hline
		$\quad\,\,\,\, a\quad \,\,\,\,$	&$\alpha $& $\alpha $&$r_1$& $r_2$ &$\alpha $ &$\alpha $\\
		\hline
		$a'$	&$1-\alpha$& $1-\alpha $&$1-r_1$& $1-r_2$ &$1-\alpha $ &$1-\alpha $\\
		\hline	
		$b$	&$u_1 $& $u_2$&$\beta$& $\beta$ &$\beta $ &$\beta $\\
		\hline	
		$b'$	&$1-u_1 $& $1-u_2 $&$1-\beta$& $1-\beta $ &$1-\beta $ &$1-\beta $\\
		\hline	
		$0_L$	&0& $0 $&$0$& $0$ &0 &0\\
		\hline	
		$1_L$	&$1$& $1 $&$1$& $1$ &1 &1\\
		\hline	
	\end{tabular}
	
\end{table}

From Table \ref{tabG-proj}, one can extract all states on $L$, related to the choice of $r_1,r_2,u_1,u_2$. Each column in the Table \ref{tabG-proj} represents a state on $L$.  As example, $m_b$ and $m_0$ are in Table \ref{tabG-proj-miery}.
\begin{table}[h!] 
	\renewcommand{\arraystretch}{1.3}
	\caption{\label{tabG-proj-miery}
	\small{States on $L$ } 
}	\centering
	\small
	\begin{tabular}{|c|c|c|c|c|c|c|} 	\hline	
		&$a$& $a'$&$b$& $b'$ &$0_L$ &$_L$\\
		\hline
		$\,\,\,m_b\,\,\,$	&$\,\,\,r_1\,\,\, $& $1-r_1 $&$\,\,\,\beta\,\,\,$& $\,1-\beta\,$ &$\,\,\,0\,\,\, $ &$\,\,\,1\,\,\, $\\
		\hline
		$m_0$	&$\alpha $& $1-\alpha $&$\beta$& $1-\beta$ &$0 $ &$1 $\\
		\hline
	\end{tabular}
\end{table}
\end{example}
\begin{definition}\label{pure} Let $G\in \Gamma_9$.  The map $G$ is called  a pure projection if   $G(a,b)=G(a,0_L)$ 
	for any $a,b\in L$.  
\end{definition}	
\begin{theorem}\label{projekcia 2}
	For every $s$-map $p$ there exists a $G$--map  $G_p\in\Gamma_9$ such that $G_p(a,b)=G_p(a,0_L)$. 
\end{theorem}
\begin{proof} Set $G_p(a,b)=p(a,b)+p(a,b')=p(a,a)$, where $p$ is an arbitrary $s$-map. Then $G_p\in\Gamma_9$ and $G_p(a,b)=G_p(a,0_L)$ for any $b\in L$. (Q.E.D.)
\end{proof}

The results for  $\Gamma_9-\Gamma_{12}$ are summarized  in Table \ref{v}.

\begin{table}[h] 
	\renewcommand{\arraystretch}{1.3}
	\centering
	\small
	\caption{\label{v}\small{Results for $\Gamma_9-\Gamma_{12}$}}
	\begin{tabular}{|c|c|c|c|c|c|c|c|c|} 	\hline	
		&$\,\Gamma_9\,$& $\,\Gamma_{10}\,$&$\,\Gamma_{11}\,$& $\,\Gamma_{12}\,$ \\
		\hline
		$\,\,$probability  of$\,\,$	&$a$& $b$&$a'$& $b'$  \\
		\hline	
	\end{tabular}
	
\end{table}

\section{\label{Summary} Summary }

Two issues related to the $G$-map on a quantum logic arised:
existence of this function and its properties on a QL.
The existence of $s$-map  has been solved in \cite{nanasiova1} and \cite{nanasiova2}. 

Some features of  $G$-map are summarized in the following:

\begin{itemize}
\item The classes $\Gamma_2$, $\Gamma_3$, $\Gamma_5$, $\Gamma_6$  are mutually isomorphic.
\item Each probability measure on $ \mathcal{B}$ induces a pseudometric. 
       It means, that for any    probability measure $m$, the map $d_m$: 
         $d_m(a,b)=m(a\wedge b')+m(a'\wedge b)$ is a pseudometric on $ \mathcal{B}$ induced by $m$.
         
	 If  $p\in\Gamma_2$  and
	 $d_p(a,b)=p(a,b')+p(a', b),$
	 then  $d_p\in\Gamma_{4}$ and it does not have to be a pseudometric on the quantum logic.
\item Let $L$ be a QL,  $m$ be a state and $p$ be an $s$-map. 
Then  the first Bell-type inequality (\ref{bem1})  is not necessarily  fulfilled and   the version  (\ref{bepl1})    is always fulfilled.
\begin{eqnarray}
\label{bem1} m(a)+m(b)~-~m(a\wedge b)&\leq &1		\\
\label{bepl1}p(a,a)+p(b,b)-p(a,b)&\leq &1
\end{eqnarray}
	The second Bell-type innequality (\ref{IIBm})  is not necessarily  fulfilled and   the version  (\ref{IIBp}) is fulfilled for every $s$-map, which induces  a pseudometric on $L$ 	 \cite{Pykacz Nan Val}.
\begin{eqnarray}
\label{IIBm} m(a)+m(b)+m(c)-m(a\wedge b)-m(a\wedge c)-m(c\wedge b)&\leq& 1\\
\label{IIBp} p(a,a)+p(b,b)+p(c,c)-p(a,b)-p(a,c)-p(c,b)&\leq &1
\end{eqnarray}	
		   
\item  Analogically implication (\ref{jpm}) (Jauch-Piron state, see e.g. \cite{jauch-piron states,honza}) can be violated on a QL $L$ but implication (\ref{jpp}) is   always valid 
\begin{eqnarray}
\label{jpm} m(a)=m(b)=1\quad\quad&\Rightarrow&\quad\quad m(a\wedge b)=1\\
\label{jpp} p(a,a)=p(b,b)=1 \quad\quad&\Rightarrow&\quad\quad p(a,b)=1,
\end{eqnarray}
	and moreover for any $c\in L$ $p(a,c)=p(c,a)=p(c,c)$.	

	 	   	\item  The classes  $\Gamma_{9}$ -- $\Gamma_{12}$:  On a Boolean algebra, every projection is a pure projection. On a quantum logic,  a $G$-map is not necessarilly a pure projection, see Example \ref{existenciaGPpr}.

\end{itemize}

The modeling of random events 	using of  quantum logics allows to test, among other things,  stochastic causality.

\end{document}